\documentclass{llncs}
\usepackage{pagr}
\usepackage{hyperref}

\title{On Infinite Guarded Recursive Specifications\\ in Process Algebra}
\author{R.J. van Glabbeek\inst{1}\fnmsep\inst{2} 
        \and
        C.A. Middelburg\inst{3}}
\institute{Data61, CSIRO, Sydney, Australia
           \and
           School of Computer Science and Engineering, 
           University of New South Wales, \\
           Sydney, Australia \\
           \email{rvg@cs.stanford.edu}
           \and
           Informatics Institute, Faculty of Science, 
           University of Amsterdam, \\
           Amsterdam, the Netherlands \\
           \email{C.A.Middelburg@uva.nl}}

\begin{document}
\maketitle

\begin{abstract}
In most presentations of \ACP\ with guarded recursion, recur\-sive 
specifications are finite or infinite sets of recursion equations of 
which the right-hand sides are guarded terms.
The completeness with respect to bisimulation equivalence of the axioms 
of \ACP\ with guarded recursion has only been proved for the special 
case where recursive specifications are finite sets of recursion 
equations of which the right-hand sides are guarded terms of a 
restricted form known as linear terms.
In this note, we widen this completeness result to the general case.
\begin{classcode}
F.1.2, F.4.1
\end{classcode}
\end{abstract}

\section{Introduction}
\label{sect-intro}

In \ACP\ with guarded recursion, guarded recursive specifications, i.e.\ 
sets of recursion equations of which the right-hand sides are guarded 
terms, are used for recursive definitions of processes 
(see e.g.~\cite{BW90}).
In most cases where \ACP\ or a variant of it is extended with guarded 
recursion, guarded recursive specifications may be infinite.
Moreover, countably infinite guarded recursive specifications are used 
in many applications of the process algebras concerned.
Nevertheless, the completeness with respect to bisimulation equivalence 
of the axioms of \ACP\ with guarded recursion has only been proved for 
the special case where recursive specifications are finite sets of 
recursion equations of which the right-hand sides are guarded terms of a 
restricted form known as linear terms.

The second author of this note realized in March 2017 that the 
completeness proof given in~\cite{Fok00} for the above-mentioned special 
case could be widened to the general case.
He communicated this at the time with several colleagues and forgot 
about it until it was recently mentioned in~\cite{Gla20a}.
This mention motivated him to write a note about the general 
completeness result.
It is noteworthy that the proof of the fact on which the widening of the 
existing completeness proof is based (Theorem~\ref{theorem-reduction}) 
turned out to be less straightforward than the second author thought in 
March 2017 and comes from the first author.

In order to make this note self-contained, it contains short surveys of 
\ACP\ and its extension with guarded recursion.
We did not attach much importance to preventing any text overlap with 
surveys from earlier papers.

\section{Algebra of Communicating Processes}
\label{sect-ACP}

In this section, we give a survey of \ACP\ (Algebra of Communicating 
Processes).
For a comprehensive overview of \ACP, the reader is referred 
to~\cite{BW90,Fok00}.

In \ACP, it is assumed that a fixed but arbitrary set $\Act$ of
\emph{actions}, with $\dead \notin \Act$, has been given.
We write $\Actd$ for $\Act \union \set{\dead}$.
It is further assumed that a fixed but arbitrary commutative and 
associative \emph{communication} function 
$\funct{\commf}{\Actd \x \Actd}{\Actd}$, with 
$\commf(\dead,a) = \dead$ for all $a \in \Actd$, has been given.
The function $\commf$ is regarded to give the result of synchronously
performing any two actions for which this is possible, and to give 
$\dead$ otherwise.

The signature of \ACP\ consists of the following constants and 
operators:
\begin{itemize}
\item
for each $a \in \Act$, the \emph{action} constant 
$a$\,;
\item
the \emph{inaction} constant $\dead$\,;
\item
the binary \emph{alternative composition} operator 
$\ph \altc \ph$\,;
\item
the binary \emph{sequential composition} operator 
$\ph \seqc \ph$\,;
\item
the binary \emph{parallel composition} operator 
$\ph \parc \ph$\,;
\item
the binary \emph{left merge} operator 
$\ph \leftm \ph$\,;
\item
the binary \emph{communication merge} operator 
$\ph \commm \ph$\,;
\item
for each $H \subseteq \Act$, the unary \emph{encapsulation} operator
$\encap{H}$\,.
\end{itemize}
We assume that there is an infinite set $\cX$ of variables which 
contains $x$, $y$ and $z$ with and without subscripts.
Terms over the signature of \ACP, also referred to as \ACP\ terms, are 
built as usual.
We use infix notation for the binary operators.
The precedence conventions used with respect to the operators of \ACP\
are as follows: $\altc$ binds weaker than all others, $\seqc$ binds
stronger than all others, and the remaining operators bind equally
strong.

The constants of \ACP\ can be explained as follows ($a \in \Act$):
\begin{itemize}
\item
$\dead$ denotes the process that cannot do anything;
\item
$a$ denotes the process that first performs action $a$ and after that 
terminates successfully.
\end{itemize}
Let $t$ and $t'$ be closed \ACP\ terms denoting processes $p$ and $p'$.
Then the operators of \ACP\ can be explained as follows:
\begin{itemize}
\item
$t \altc t'$ denotes the process that behaves as $p$ or behaves as $p'$ 
(but not both);
\item
$t \seqc t'$\, denotes the process that first behaves as $p$ and on 
successful termina\-tion \nolinebreak[2] of $p$ next behaves as $p'$;
\item
$t \parc t'$ denotes the process that behaves as $p$ and $p'$ in 
parallel;
\item
$t \leftm t'$ denotes the same process as $t \parc t'$, except that it 
starts with performing an action of $p$;
\item
$t \commm t'$ denotes the same process as $t \parc t'$, except that it 
starts with performing an action of $p$ and an action of $p'$ 
synchronously;
\item
$\encap{H}(t)$ denotes the process that behaves the same as $p$, except 
that actions from $H$ are blocked.
\end{itemize}
The operators $\leftm$ and $\commm$ are of an auxiliary nature.
They are needed to axiomatize \ACP.

The axioms of \ACP\ are the equations given in Table~\ref{axioms-ACP}.
\begin{table}[!t]
\caption{Axioms of \ACP}
\label{axioms-ACP}
\begin{eqntbl}
\begin{axcol}
x \altc y = y \altc x                                  & \axiom{A1}   \\
(x \altc y) \altc z = x \altc (y \altc z)              & \axiom{A2}   \\
x \altc x = x                                          & \axiom{A3}   \\
(x \altc y) \seqc z = x \seqc z \altc y \seqc z        & \axiom{A4}   \\
(x \seqc y) \seqc z = x \seqc (y \seqc z)              & \axiom{A5}   \\
x \altc \dead = x                                      & \axiom{A6}   \\
\dead \seqc x = \dead                                  & \axiom{A7}   \\
{}                                                                    \\
\encap{H}(a) = a                \hfill \mif a \notin H & \axiom{D1}   \\
\encap{H}(a) = \dead            \hfill \mif a \in H    & \axiom{D2}   \\
\encap{H}(x \altc y) = \encap{H}(x) \altc \encap{H}(y) & \axiom{D3}   \\
\encap{H}(x \seqc y) = \encap{H}(x) \seqc \encap{H}(y) & \axiom{D4}  
\end{axcol}
\qquad \qquad
\begin{axcol}
x \parc y =
          x \leftm y \altc y \leftm x \altc x \commm y & \axiom{CM1}  \\
a \leftm x = a \seqc x                                 & \axiom{CM2}  \\
a \seqc x \leftm y = a \seqc (x \parc y)               & \axiom{CM3}  \\
(x \altc y) \leftm z = x \leftm z \altc y \leftm z     & \axiom{CM4}  \\
a \seqc x \commm b = (a \commm b) \seqc x              & \axiom{CM5}  \\
a \commm b \seqc x = (a \commm b) \seqc x              & \axiom{CM6}  \\
a \seqc x \commm b \seqc y =
                       (a \commm b) \seqc (x \parc y)  & \axiom{CM7}  \\
(x \altc y) \commm z = x \commm z \altc y \commm z     & \axiom{CM8}  \\
x \commm (y \altc z) = x \commm y \altc x \commm z     & \axiom{CM9}  \\
{}                                                                    \\
{}                                                                    \\
a \commm b = \commf(a,b)                               & \axiom{CF}  
\end{axcol}
\end{eqntbl}
\end{table}
In these equations, $a$ and $b$ stand for arbitrary constants of \ACP, 
and $H$ stands for an arbitrary subset of $\Act$.
In D1 and D2, side conditions restrict what $a$ and $H$ stand for.

In the sequel, we will use the sum notation $\Altc{i<n} t_i$.
Let $t_0, t_1, t_2, \ldots$ be terms over the signature of \ACP\ or an 
extension of \ACP.
Then $\Altc{i<0} t_i = \dead$ and, for each $n \in \Nat$ with $n > 0$,
the term $\Altc{i<n} t_i$ is defined by induction on $n$ as follows:
$\Altc{i<1} t_i = t_0$ and $\Altc{i<n+1} t_i =\Altc{i<n} t_i \altc t_n$.

\section{\ACP\ with Guarded Recursion}
\label{sect-ACPrec}

In this section, we give a survey of the extension of \ACP\ with guarded 
recursion.
For a comprehensive overview of guarded recursion in the setting of \ACP, 
the reader is referred to~\cite{BW90,Fok00}.

A closed \ACP\ term denotes a process with a finite upper bound to the 
number of actions that it can perform. 
Guarded recursion allows the description of processes without a finite 
upper bound to the number of actions that it can perform.

Let $t$ be a term over the signature of \ACP\ or an extension of \ACP\ 
in which a variable $X$ occurs.
Then an occurrence of $X$ in $t$ is \emph{guarded} if $t$ has a subterm 
of the form $a \seqc t'$ where $a \in \Act$ and $t'$ contains this 
occurrence of $X$.
An \ACP\ term $t$ is a guarded \ACP\ term if all occurrences of 
variables in $t$ are guarded.

A \emph{guarded recursive specification} over \ACP\ is a set 
$\set{X_i = t_i \where i \in I}$, 
where $I$ is a finite or infinite set, 
each $X_i$ is a variable from $\cX$, 
each $t_i$ is either a guarded \ACP\ term in which only variables from 
$\set{X_i \where i \in I}$ occur or an \ACP\ term rewritable to such a 
term using the axioms of \ACP\ in either direction and/or the equations 
in $\set{X_j = t_j \where j \in I \Land i \neq j}$ from left to right, 
and $X_i \neq X_j$ for all $i,j \in I$ with $i \neq j$.

We write $\vars(E)$, where $E$ is a guarded recursive specification, for 
the set of all variables that occur in $E$.
The equations occurring in a guarded recursive specification are called 
\emph{recursion equations}.

A solution of a guarded recursive specification $E$ in some model of 
\ACP\ is a set $\set{p_X \where X \in \vars(E)}$ of elements of the 
carrier of that model such that each equation in $E$ holds if, for all 
$X \in \vars(E)$, $X$ is assigned $p_X$.
We are only interested in models of \ACP\ in which guarded recursive 
specifications have unique solutions. 

We extend \ACP\ with guarded recursion by adding constants for solutions 
of guarded recursive specifications over \ACP\ and axioms concerning 
these additional constants.
For each guarded recursive specification $E$ over \ACP\ and each 
$X \in \vars(E)$, we add a constant $\rec{X}{E}$ that stands for the 
unique solution of $E$ for $X$ to the constants of \ACP.
We add the equation RDP and the conditional equation RSP given in 
Table~\ref{axioms-REC} to the axioms of \ACP.
\begin{table}[!t]
\caption{Axioms for guarded recursion}
\label{axioms-REC}
\begin{eqntbl}
\begin{saxcol}
\rec{X}{E} = \rec{t}{E} & \mif X = t\; \in \;E          & \axiom{RDP} \\
E \Limpl X = \rec{X}{E}   & \mif X \in \vars(E)         & \axiom{RSP} 
\end{saxcol}
\end{eqntbl}
\end{table}
In RDP and RSP, $X$ stands for an arbitrary variable from $\cX$, $t$ 
stands for an arbitrary \ACP\ term, $E$ stands for an arbitrary 
guarded recursive specification over \ACP, and the notation $\rec{t}{E}$ 
is used for $t$ with, for all $X \in \vars(E)$, all occurrences of $X$ 
in $t$ replaced by $\rec{X}{E}$.
Side conditions restrict what $X$, $t$ and $E$ stand for.
We write $\ACPr$ for the resulting theory.
Terms over the signature of \ACPr\ are also referred to as \ACPr\ terms.

The equations $\rec{X}{E} = \rec{t}{E}$ and the conditional equations 
\mbox{$E \Limpl X \!=\! \rec{X}{E}$} for a fixed $E$ express that the 
constants $\rec{X}{E}$ make up a solution of $E$ and that this solution 
is the only one.

Because we have to deal with conditional equational formulas with an 
infinite number of premises in \ACPr, it is understood that infinitary 
conditional equational logic is used in deriving equations from the 
axioms of \ACPr.
A complete inference system for infinitary conditional equational logic 
can be found in~\cite{GV93}.
It is noteworthy that in the case of infinitary conditional equational 
logic derivation trees may be infinitely branching (but they may not 
have infinite branches).

We write $T \Ent t = t'$, where $T$ is \ACP\ or \ACPr, to indicate that 
the equation $t = t'$ is derivable from the axioms of $T$ using a 
complete inference system for infinitary conditional equational logic.

\section{Linear Recursive Specifications}
\label{sect-linear-ACPrec}

In this section, we show that each guarded recursive specification over 
\ACP\ can be reduced to one in which the right-hand sides of recursion 
equations are guarded terms of a restricted form known as linear terms.
This result will be used in Section~\ref{sect-complete-ACPrec}.
In its proof, we make use of the fact that each guarded \ACP\ term has a 
head normal form.

Let $T$ be \ACP\ or \ACPr.
The set $\HNF$ of \emph{head normal forms of $T$} is inductively 
defined by the following rules:
\begin{itemize}
\item 
$\dead \in \HNF$;
\item 
if $a \in \Act$, then $a \in \HNF$;
\item 
if $a \in \Act$ and $t$ is a term over the signature of $T$, then 
$a \seqc t \in \HNF$;
\item 
if $t,t' \in \HNF$, then $t \altc t '\in \HNF$.
\end{itemize}
Each head normal form of $T$ is derivably equal to a head normal form 
of the form $\Altc{i < n} a_i \seqc t_i \altc \Altc{j < m} b_i$, where 
$n,m \in \Nat$, for each $i < n$, $a_i \in \Act$ and $t_i$ is a term over 
the signature of $T$, and, for each $j < m$, $b_j \in \Act$. 

Each guarded \ACPr\ term is derivably equal to a head normal form of 
\ACPr.
\begin{proposition}[Head normal form]
\label{prop-HNF-ACP}
For each guarded \ACPr\ term $t$, there exists a head normal form $t'$ 
of \ACPr\ such that $\ACPr \Ent t = t'$.
\end{proposition}
\begin{proof}
First we prove the following weaker result about head normal forms:
\begin{quote}
\emph{For each guarded \ACP\ term $t$, there exists a head normal form 
$t'$ of \ACP\ such that $\ACP \Ent t = t'$}.
\end{quote}
The proof is straightforward by induction on the structure of $t$.
The case where $t$ is of the form $\dead$ and the case where $t$ is of 
the form $a$ ($a \in \Act$) are trivial.
The case where $t$ is of the form $t_1 \altc t_2$ follows immediately 
from the induction hypothesis.
The case where $t$ is of the form $t_1 \seqc t_2$ follows immediately 
from the induction hypothesis and the claim that, for all head normal 
forms $t_1$ and $t_2$ of \ACP, there exists a head normal form $t'$ of 
\ACP\ such that $t_1 \seqc t_2 = t'$ is derivable from the axioms of 
\ACP.
This claim is easily proved by induction on the structure of~$t_1$.
The cases where $t$ is of one of the forms $t_1 \leftm t_2$, 
$t_1 \commm t_2$ or $\encap{H}(t_1)$ are proved along the same lines as 
the case where $t$ is of the form $t_1 \seqc t_2$.
In the case that $t$ is of the form $t_1 \commm t_2$, each of the cases 
to be considered in the inductive proof of the claim demands a proof by 
induction on the structure of~$t_2$.
The case that $t$ is of the form $t_1 \parc t_2$ follows immediately 
from the case that $t$ is of the form $t_1 \leftm t_2$ and the case that 
$t$ is of the form $t_1 \commm t_2$.
Because~$t$ is a guarded \ACP\ term, the case where $t$ is a variable 
cannot occur.

The proof of the proposition itself is also straightforward by induction 
on the structure of $t$.
The cases other than the case where $t$ is of the form $\rec{X}{E}$ is 
proved in the same way as in the above proof of the weaker result.
The case where $t$ is of the form $\rec{X}{E}$ follows immediately from
the weaker result and RDP.
\qed
\end{proof}

The set $\LT$ of \emph{linear \ACP\ terms} is inductively defined by 
the following rules:
\begin{itemize}
\item
$\dead \in \LT$;
\item
if $a \in \Act$, then $a \in \LT$;
\item
if $a \in \Act$ and $X \in \cX$, then $a \seqc X \in \LT$;
\item
if $t,t' \in \LT$, then $t \altc t' \in \LT$.
\end{itemize}
Clearly, each linear \ACP\ term is also a guarded \ACP\ term (but not 
vice versa).

A \emph{linear recursive specification} over \ACP\ is a guarded
recursive specification $E$ over \ACP\ such that, for each equation 
$X = t \,\in\, E$, $t \in \LT$.

Each guarded recursive specification over \ACP\ can be reduced to a
linear recursive specification over \ACP. 
\begin{theorem}[Reduction]
\label{theorem-reduction}
\sloppy
For each guarded recursive specification $E$ over \ACP\ and each 
$X \in \vars(E)$, there exists a finite or countably infinite linear 
recursive specification $E'$ over \ACP\ such that 
$\ACPr \Ent \rec{X}{E} = \rec{X}{E'}$.
\end{theorem}
\begin{proof}
We approach this algorithmically.
In the construction of the linear recursive specification $E'$, we keep
a set $V$ of recursion equations from $E'$ that are already found and a
sequence $W$ of equations of the form $X_k = \rec{t_k}{E}$ that still 
have to be transformed.
The algorithm has a finite or countably infinite number of stages.
In each stage, $V$ and $W$ are finite.
Initially, $V$ is empty and $W$ contains only the equation 
$X_0 = \rec{X}{E}$.

In each stage, we remove the first equation from $W$.
Assume that this equation is $X_k = \rec{t_k}{E}$. 
We bring the term $\rec{t_k}{E}$ into head normal form. 
If $t_k$ is not a guarded term, then we use RDP here to turn $t_k$ into 
a guarded term first.
Thus, by Proposition~\ref{prop-HNF-ACP}, we can always bring the term 
$\rec{t_k}{E}$ into head normal form.
Assume that the resulting head normal form is
$\Altc{i<n} a_i \seqc t'_i \altc \Altc{j<m} b_j$.
Then, we add the equation 
$X_k = \Altc{i<n} a_i \seqc X_{k+i+1} \altc  \Altc{j<m} b_j$,
where the $X_{k+i+1}$ are fresh variables, to the set $V$.
Moreover, for each $i < n$, we add the equation $X_{k+i+1} = t'_i$ to 
the end of the sequence $W$.
Notice that the terms $t'_i$ are of the form $\rec{t_{k+i+1}}{E}$.

Because $V$ grows monotonically, there exists a limit. 
That limit is the finite or countably infinite linear recursive 
specification $E'$.
Every equation that is added to the finite sequence $W$, is also removed 
from it.
Therefore, the right-hand side of each equation from $E'$ only contains
variables that also occur as the left-hand side of an equation from 
$E'$.

Now, we want to use RSP to show that 
$\ACPr \Ent \rec{X}{E} = \rec{X}{E'}$.
The variables occurring in $E'$ are $X_0, X_1, X_2, \ldots\;$.
For each $k$, the variable $X_k$ has been exactly once in $W$ as the 
left-hand side of an equation.
For each $k$, assume that  this equation is $X_k = \rec{t_k}{E}$.
To use RSP, we have to show for each $k$ that the equation
$X_k = \Altc{i<n} a_i \seqc X_{k+i+1} \altc  \Altc{j<m} b_j$ from $E'$
with, for each $l$, all occurrences of $X_l$ replaced by $\rec{t_l}{E}$
is derivable from the axioms of \ACPr.
For each $k$, this follows from the construction.
\qed
\end{proof}
An immediate corollary of Theorem~\ref{theorem-reduction} is the 
following expressiveness result: \linebreak[2]
in each model of \ACPr, the processes that can be described by a guarded 
recursive specification over \ACP\ and the processes that can be 
described by a finite or countably infinite linear recursive 
specification over \ACP\ are the same.

\section{Semantics of \ACP\ with Guarded Recursion}
\label{sect-semantics-ACPr}

In this section, we present a structural operational semantics of 
\ACPr\ and define a notion of bisimulation equivalence based on this 
semantics.

We start with presenting a structural operational semantics of \ACPr.
The following relations on closed \ACPr\ terms are used:
\begin{itemize}
\item 
for each $a \in \Act$, a unary relation $\aterm{}{a}$\,;
\item 
for each $a \in \Act$, a binary relation $\astep{}{a}{}$\,.
\end{itemize}
We write $\aterm{t}{a}$ instead of ${\aterm{}{a}}\,(t)$ and
$\astep{t}{a}{t'}$ instead of ${\astep{}{a}{}}\,(t,t')$.
The relations $\aterm{}{a}$ and $\astep{}{a}{}$ can be explained as 
follows:
\begin{itemize}
\item
$\aterm{t}{a}$: 
$t$ can perform action $a$ and then terminate successfully;
\item
$\astep{t}{a}{t'}$: 
$t$ can perform action $a$ and then behave as $t'$.
\end{itemize}
The structural operational semantics of \ACPr\ is described by the
rules given in Table~\ref{rules-ACPr}.
\begin{table}[!t]
\caption{Rules for the operational semantics of \ACPr}
\label{rules-ACPr}
\begin{ruletbl}
\Rule
{}
{\aterm{a}{a}}
\qquad
\\
\Rule
{\aterm{x}{a}}
{\aterm{x \altc y}{a}}
\qquad
\Rule
{\aterm{y}{a}}
{\aterm{x \altc y}{a}}
\qquad
\Rule
{\astep{x}{a}{x'}}
{\astep{x \altc y}{a}{x'}}
\qquad
\Rule
{\astep{y}{a}{y'}}
{\astep{x \altc y}{a}{y'}}
\\
\Rule
{\aterm{x}{a}}
{\astep{x \seqc y}{a}{y}}
\qquad
\Rule
{\astep{x}{a}{x'}}
{\astep{x \seqc y}{a}{x' \seqc y}}
\\
\Rule
{\aterm{x}{a}}
{\astep{x \parc y}{a}{y}}
\qquad
\Rule
{\aterm{y}{a}}
{\astep{x \parc y}{a}{x}}
\qquad
\Rule
{\astep{x}{a}{x'}}
{\astep{x \parc y}{a}{x' \parc y}}
\qquad
\Rule
{\astep{y}{a}{y'}}
{\astep{x \parc y}{a}{x \parc y'}}
\\
\RuleC
{\aterm{x}{a},\; \aterm{y}{b}}
{\aterm{x \parc y}{c}}
{\commf(a,b) = c}
\qquad
\RuleC
{\aterm{x}{a},\; \astep{y}{b}{y'}}
{\astep{x \parc y}{c}{y'}}
{\commf(a,b) = c}
\\
\RuleC
{\astep{x}{a}{x'},\; \aterm{y}{b}}
{\astep{x \parc y}{c}{x'}}
{\commf(a,b) = c}
\qquad
\RuleC
{\astep{x}{a}{x'},\; \astep{y}{b}{y'}}
{\astep{x \parc y}{c}{x' \parc y'}}
{\commf(a,b) = c}
\\
\Rule
{\aterm{x}{a}}
{\astep{x \leftm y}{a}{y}}
\qquad
\Rule
{\astep{x}{a}{x'}}
{\astep{x \leftm y}{a}{x' \parc y}}
\\
\RuleC
{\aterm{x}{a},\; \aterm{y}{b}}
{\aterm{x \commm y}{c}}
{\commf(a,b) = c}
\qquad
\RuleC
{\aterm{x}{a},\; \astep{y}{b}{y'}}
{\astep{x \commm y}{c}{y'}}
{\commf(a,b) = c}
\\
\RuleC
{\astep{x}{a}{x'},\; \aterm{y}{b}}
{\astep{x \commm y}{c}{x'}}
{\commf(a,b) = c}
\qquad
\RuleC
{\astep{x}{a}{x'},\; \astep{y}{b}{y'}}
{\astep{x \commm y}{c}{x' \parc y'}}
{\commf(a,b) = c}
\\
\RuleC
{\aterm{x}{a}}
{\aterm{\encap{H}(x)}{a}}
{a \not\in H}
\qquad
\RuleC
{\astep{x}{a}{x'}}
{\astep{\encap{H}(x)}{a}{\encap{H}(x')}}
{a \not\in H}
\\
\RuleC
{\aterm{\rec{t}{E}}{a}}
{\aterm{\rec{X}{E}}{a}}
{X = t\; \in \;E}
\qquad
\RuleC
{\astep{\rec{t}{E}}{a}{x'}}
{\astep{\rec{X}{E}}{a}{x'}}
{X = t\; \in \;E}
\end{ruletbl}
\end{table}
In these tables, $a$, $b$, and $c$ stand for arbitrary actions from 
$\Act$, 
$X$~stands for an arbitrary variable from $\cX$, 
$t$~stands for an arbitrary \ACP\ term, and 
$E$~stands for an arbitrary guarded recursive specification over \ACP.

A \emph{bisimulation} is a binary relation $R$ on closed \ACPr\ terms 
such that, for all closed \ACPr\ terms $t_1,t_2$ with $R(t_1,t_2)$, the 
following conditions hold: 
\begin{itemize}
\item
if $\astep{t_1}{a}{t_1'}$, then there exists a closed \ACPr\ term $t_2'$ 
such that $\astep{t_2}{a}{t_2'}$ and~$R(t_1',t_2')$;
\item
if $\astep{t_2}{a}{t_2'}$, then there exists a closed \ACPr\ term $t_1'$ 
such that $\astep{t_1}{a}{t_1'}$ and~$R(t_1',t_2')$;
\item
$\aterm{t_1}{a}$ iff $\aterm{t_2}{a}$.
\end{itemize}%
Two closed \ACPr\ terms $t_1,t_2$ are \emph{bisimulation equivalent}, 
written $t_1 \bisim t_2$, if there exists a bisimulation $R$ such that 
$R(t_1,t_2)$.

\begin{proposition}[Congruence]
\label{prop-prob-bisim-congr}
${} \bisim {}$ is a congruence with respect to the operators of \ACPr.
\end{proposition}
The axioms of \ACPr\ are sound with respect to bisimulation equivalence 
for equations between closed terms.
\begin{theorem}[Soundness]
\label{theorem-soundness-ACPr}
For all closed \ACPr\ terms $t$ and $t'$, $t \bisim t'$ if 
$\ACPr \Ent t = t'$.
\end{theorem}
The proofs of Proposition~\ref{prop-prob-bisim-congr} and 
Theorem~\ref{theorem-soundness-ACPr} can, for example, be found 
in~\cite{BW90}.

\section{Completeness of \ACP\ with Guarded Recursion}
\label{sect-complete-ACPrec}

It follows from Theorem~\ref{theorem-reduction} and the completeness 
proof given in~\cite{Fok00} for the special case of finite linear 
recursive specifications over \ACP\ that the axioms of \ACPr\ are also 
complete with respect to bisimulation equivalence for equations between 
closed terms.
\begin{theorem}[Completeness]
\label{theorem-completeness}
For all closed \ACPr\ terms $t$ and $t'$, $t \bisim t'$ only if
$\ACPr \Ent t = t'$.
\end{theorem}
\begin{proof}
Theorem~4.4.1 from~\cite{Fok00} states that, for all closed \ACPr\ 
terms $t$ and $t'$ in which only constants $\rec{X}{E}$ occur where $E$ 
is a finite linear recursive specification, $t \bisim t'$ only if 
$\ACPr \Ent t = t'$.
We can strengthen this theorem by dropping the finiteness condition 
because the proof given in~\cite{Fok00} does not rely on it.
It follows immediately from the strengthened version of Theorem~4.4.1 
from~\cite{Fok00} and Theorem~\ref{theorem-reduction} from the current 
paper that, for all closed \ACPr\ terms $t$ and $t'$, $t \bisim t'$ only 
if $\ACPr \Ent t = t'$.
\qed
\end{proof}
To the best of our knowledge, the completeness of the axioms of \ACPr\ 
with respect to bisimulation equivalence has as yet only been proved for 
the special case of finite linear recursive specifications.
Crucial for the completeness for the general case is that infinitary 
conditional equational logic is used in deriving equations from the 
axioms of \ACPr.
The use of this logic is inescapable with infinite guarded recursive 
specifications.
This speaks for itself, but it is virtually unmentioned in the 
literature on process algebra.

\section{Concluding Remarks}
\label{sect-concl}

We have widened the existing completeness result for \ACPr. 
A by-product of this work is the following expressiveness result: in 
each model of \ACPr, the processes that can be described by a guarded 
recursive specification over \ACP\ and the processes that can be 
described by a finite or countably infinite linear recursive 
specification over \ACP\ are the same.
Notice that even uncountably infinite guarded recursive specifications 
over \ACP\ can be reduced to finite or countably infinite linear 
recursive specifications over \ACP.

\bibliographystyle{splncs03}
\bibliography{PA}

\end{document}